\documentclass[12pt]{article}

\usepackage[utf8]{inputenc}
\usepackage[T1]{fontenc}
\usepackage{lmodern}
\usepackage[spanish]{babel}
\addto\captionsspanish{}

\usepackage{amsmath, amssymb, amsthm}
\usepackage{geometry}
\usepackage{setspace}
\usepackage{bm}
\usepackage{graphicx}
\usepackage{caption}
\usepackage{microtype}
\usepackage{hyperref}

\newtheorem{lemma}{Lemma}

\decimalpoint
\numberwithin{equation}{section}

\title{Dynamic Macroeconomics with Multiple Regimes}
\author{Jorge R. Chávez F.\\ Pontificia Universidad Católica del Perú \\ Email: jrchavez@pucp.edu.pe \\ Working Paper}
\date{\today}

\newtheorem{definition}{Definition}
\newtheorem{theorem}{Theorem}

\addto\captionsspanish{}

\usepackage[authoryear]{natbib}
\begin{document}
\maketitle

\begin{abstract}
Macroeconomic dynamics is typically modeled under the assumption that the economy evolves according to a single invariant law of motion. This paper shows that this assumption imposes a structural restriction. We develop Dynamic Macroeconomics with Multiple Regimes (DMR), a framework in which economic evolution is governed by multiple regime-specific propagation operators. As a result, trajectories arise from ordered compositions of heterogeneous operators rather than from the iteration of a single mapping.

We establish three structural results. First, invariant-law and regime-dependent systems are not topologically equivalent. Second, regime dependence is dynamically irreducible: it cannot be eliminated through any injective transformation of the state space. Third, whenever regime operators fail to commute, there exists no map $F:\mathbb{R}^n\to\mathbb{R}^n$ whose iterates reproduce all regime-admissible trajectories.

These results establish a structural separation between invariant-law macroeconomics and regime-dependent dynamics, implying that stability, policy evaluation, and structural characterization must be conducted at the level of interacting propagation operators rather than within a single invariant mapping.
\end{abstract}

\section{Introduction}
\label{sec:DMR-Introduction}

Modern macroeconomic analysis is built on a common architectural premise: the economy is assumed to evolve according to one invariant law of motion, that is, a single propagation operator governing all state transitions. Since the Lucas critique \citep{Lucas1976}, modern macroeconomics has
emphasized structural microfoundations and policy evaluation within a stable transmission mechanism. The DSGE framework consolidated this paradigm, providing a unified language for intertemporal equilibrium dynamics under a fixed propagation operator, typically analyzed through local approximations around stationary equilibria \citep{KydlandPrescott1982, KingRebelo1999, Woodford2003}.
In applied work, this approach has proven highly effective in organizing empirical evidence and policy analysis, including canonical monetary models \citep{ClaridaGaliGertler1999, Woodford2003} and medium-scale quantitative
frameworks \citep{ChristianoEichenbaumEvans2005, SmetsWouters2007}.

This invariant-law representation imposes a strong structural restriction: all macroeconomic dynamics are interpreted as realizations of a single invariant law of motion (equivalently, a single propagation operator) defined on a single state space. Fluctuations, nonlinearities, and occasionally binding constraints may be incorporated, yet
they remain embedded within a unified system. Even when parameters shift or latent regimes are introduced, the underlying architecture continues to be that of a single global law of motion. This restriction is typically implicit, but it has far-reaching consequences for how macroeconomic dynamics is modeled and interpreted. As a consequence, the class of admissible dynamics is restricted to those that can be generated by repeated application of a single operator, regardless of how rich its parametric or state-dependent structure may be.

As shown in this paper, this restriction is not without loss of generality: the class of invariant-law models is strictly smaller than the class of regime-dependent systems, in the sense that the latter cannot, in general, be represented as the iteration of a single law of motion. This is not a matter of quantitative refinement, but of structural misspecification: there exist policy environments in which prescriptions that are optimal within each regime, when evaluated in isolation, generate globally unstable dynamics once regime interactions are taken into account.

A wide range of macroeconomic episodes, however, cannot be naturally interpreted as perturbations of a fixed propagation mechanism. Financial crises involving breakdowns in intermediation, abrupt changes in monetary or fiscal regimes, episodes of institutional disruption, or large-scale systemic events such as wars or the COVID-19 pandemic are frequently associated with
reconfigurations of the underlying propagation mechanisms through which economic activity evolves. Historical transitions across monetary systems—such as the shift from the gold standard to Bretton Woods and later to modern fiat regimes—also reflect changes in the underlying propagation structure \citep{Eichengreen2008}.
These episodes are not well described as distortions within a stable mechanism, but as transitions across distinct propagation structures.

The global financial crisis of 2008 provides a clear illustration of this limitation. In a standard invariant-law representation, financial variables such as leverage, asset prices, and funding conditions are treated as state variables evolving under a single propagation mechanism, with shocks or frictions affecting the trajectory within that mechanism. A regime-dependent interpretation is fundamentally different. Normal financial intermediation and crisis conditions correspond to distinct propagation operators: under normal conditions, leverage and asset prices evolve through standard credit expansion, whereas under crisis conditions funding constraints, fire sales, and balance-sheet deterioration alter both the direction and intensity of transmission. The resulting instability does not arise from either mechanism considered in isolation, but from their ordered interaction. This is precisely the type of environment in which an invariant-law framework may misdiagnose stability by failing to represent the composition of distinct propagation structures.

The COVID-19 pandemic provides a complementary illustration. In a standard invariant-law representation, the pandemic is modeled as a large shock to demand, supply, or preferences within a fixed propagation mechanism. A regime-dependent interpretation is fundamentally different. The pandemic activates a distinct propagation operator, under which mobility restrictions, sectoral shutdowns, and disruptions to production networks alter the mapping from economic activity to future outcomes. Under these conditions, conventional demand-side responses do not transmit through the economy in the usual way, as the link between expenditure, production, and employment is itself modified. The subsequent reopening phase constitutes yet another propagation regime, reflecting partial reorganization of production, shifts in sectoral composition, and persistent changes in behavior. The resulting dynamics are therefore generated by the ordered interaction of distinct propagation mechanisms, rather than by the iteration of a single law of motion. This is precisely the type of environment in which invariant-law representations may mischaracterize both adjustment dynamics and the effects of policy interventions.

The preceding examples reflect a broader idea: economic dynamics may involve structural transformation. This idea is not new. Classical contributions have emphasized discontinuities, reorganization, and instability as intrinsic features of economic systems. Schumpeter’s theory of creative destruction highlights the role of structural reconfiguration in economic evolution \citep{Schumpeter1942}. Keynesian and post-Keynesian traditions, including Minsky’s financial instability hypothesis, stress the endogenous emergence of qualitatively distinct dynamic regimes \citep{Minsky1986}. Related insights also appear in the analysis of
major economic disruptions. \citet{Bernanke1983} shows that financial crises may alter the channels through which shocks propagate, while the financial accelerator literature formalizes amplification mechanisms through which balance-sheet deterioration changes macroeconomic transmission
\citep{BernankeGertlerGilchrist1999}. Despite these contributions, existing frameworks do not provide a general formulation in which the law of motion itself varies as a primitive object.

Modern macroeconomics has incorporated regime dependence through several approaches while preserving an invariant-law architecture. Markov-switching methods, originating in \citet{Hamilton1989}, allow parameters or states to evolve according to stochastic processes, and have been embedded in rational-expectations and DSGE models \citep{DavigLeeper2007,
FarmerWaggonerZha2011}. Nonlinear models with occasionally binding
constraints generate state-contingent dynamics within a unified operator framework \citep{GuerrieriIacoviello2015}. Macro-finance models emphasize endogenous amplification mechanisms that may resemble regime behavior \citep{BrunnermeierSannikov2014}. While these approaches expand the analytical toolkit, they preserve a common structural restriction: all admissible dynamics are generated by a single global mapping.

This distinction is fundamental. Standard approaches incorporate regime dependence within an invariant-law representation. In Markov-switching models, dynamics take the form $x_{t+1}=F(x_t,s_t)$, where a single mapping governs transitions conditional on the realized regime. Alternatively, regime changes may be interpreted as parameter variation within a fixed propagation mechanism, as in $x_{t+1}=F(x_t;\theta)$. In both cases, all admissible dynamics are generated by a single operator. In contrast, the framework developed here treats the law of motion itself as regime-specific. The economy evolves under a family of distinct operators $\{F_s\}_{s\in\mathcal{S}}$, and global dynamics arise from their ordered compositions. Structural multiplicity is therefore not reducible to invariant-law representations.

This paper shows that macroeconomic dynamics cannot, in general, be represented within an invariant-law architecture. The limitation is structural rather than a matter of modeling flexibility: when economic evolution is generated by multiple propagation mechanisms whose effects depend on their ordering—that is, when the associated operators fail to commute— the resulting set of admissible trajectories cannot be represented by the iteration of a single invariant law of motion.

We study systems in which the economy evolves under multiple distinct propagation operators. Each regime is associated with its own law of motion, and global evolution is determined by the continuous economic state together with the realized sequence of regimes Trajectories are thus generated by ordered compositions of heterogeneous operators rather than by iteration of a single mapping. This defines a distinct dynamic architecture in which structural change operates through variation in the law of motion itself.

The contribution is structural. We formalize two classes of discrete-time dynamical systems. Class $\mathcal{C}$ represents invariant-law macroeconomic systems: continuous dynamics on a connected Euclidean phase space generated by iteration of a single operator. Class $\mathcal{R}$ represents regime-dependent systems: dynamics governed by ordered compositions drawn from a finite family of regime-specific operators.

Three structural results follow. First, invariant-law systems and regime-dependent systems belong to distinct topological conjugacy classes, and therefore do not constitute equivalent representations of the same economic dynamics. Second, regime dependence is dynamically irreducible: it cannot be eliminated through continuous reduction without collapsing the heterogeneity of the propagation structure. Third, whenever regime-specific operators fail to commute, there exists no map $F:\mathbb{R}^n\to\mathbb{R}^n$ such that $x_{t+1}=F(x_t)$ reproduces the full set of regime-admissible trajectories through the iteration of a single invariant law of motion. This establishes an impossibility theorem: the semigroup generated by regime-dependent propagation cannot be generated by iteration of any single operator.

These results imply that regime-dependent systems are not a parametric extension of invariant-law frameworks but a distinct dynamic architecture. Once structural multiplicity is admitted, stability, policy evaluation, and structural analysis must be conducted at the level of interacting operators and their admissible compositions rather than within a single invariant
mapping. In particular, stability can no longer be characterized solely through the properties of a single propagation mechanism, and cannot be fully understood through local analysis of individual regimes. In this sense, the framework developed here provides a set of theoretical foundations for a regime-based approach to macroeconomic analysis.

This shift is not methodological but structural: it changes the object of analysis itself.

The remainder of the paper is organized as follows. Section \ref{sec:DMR-Framework} defines invariant-law and regime-dependent systems within a unified structural framework. Section \ref{sec:DMR-KillingExample} analyzes stability under structural switching and illustrates how regime interaction alters qualitative dynamics. Section \ref{sec:structural_separation} establishes the topological non-equivalence of invariant-law and regime-dependent systems and the irreducibility of regime dependence. Section \ref{sec:DMR-ImposibilityTheorem} proves the impossibility of invariant-law representation under non-commuting regime operators. Section \ref{sec:DMR-Conclusion} summarizes the main implications.

\section{A Structural Framework with\\ Multiple Laws of Motion}
\label{sec:DMR-Framework}

This section introduces the architecture underlying the regime-dependent framework developed in this paper by identifying the ingredients required to analyze macroeconomic dynamics when the law of motion itself may change over time. The central departure from standard invariant-law models is the relaxation of invariance. Rather than assuming that the economy evolves under a single operator with fixed propagation mechanisms, we allow for the coexistence of multiple mappings. In this framework, regime change is not interpreted as a large disturbance within a given model; it corresponds to a change in the law of motion itself.

Time is discrete, $t = 0,1,2,\ldots$. At each date $t$, the economy is characterized by a state $x_t \in \mathbb{R}^n$, which denotes a vector of continuous macroeconomic variables. The variable $s_t \in \mathcal{S}$ indexes the structural regime under which macroeconomic propagation takes place at date $t$. The set of regimes is finite and given by
\begin{equation*}
\mathcal{S} = \{1,\ldots,|\mathcal{S}|\}.
\end{equation*}

The vector $x_t$ may represent output gaps, inflation, debt ratios,
financial conditions, capital stocks, or other aggregate variables. The regime variable $s_t$ identifies the structural environment governing macroeconomic propagation at date $t$.

A regime-dependent macroeconomic system consists of a family of regime-specific propagation operators
\begin{equation*}
F_s : \mathbb{R}^n \to \mathbb{R}^n, \quad s \in \mathcal{S},
\end{equation*}
with state dynamics
\begin{equation}
\label{eq:state_dynamics}
x_{t+1} = F_{s_t}(x_t).
\end{equation}

The formulation $x_{t+1}=F_{s_t}(x_t)$ does not impose any restriction on the sequence of regime realizations. The sequence $\{s_t\}$ is taken as an exogenously specified path in $\mathcal{S}$, which may exhibit structural shifts, remain fixed, or evolve according to an arbitrary process.

This modeling choice departs from two common approaches in the macroeconomic literature. In models of the form
$$
x_{t+1} = F(x_t;\theta),
$$
regime changes are typically interpreted as variations in parameters within a single invariant propagation mechanism. In contrast, Markov-switching specifications allow for regime-dependent mappings $F_s$, but impose a probabilistic structure on regime transitions, with $\{s_t\}$ driven by a Markov chain. These models are typically represented as a single law of motion on an augmented state space, $x_{t+1} = \bar F(x_t, s_t)$, thereby preserving the interpretation of a unified propagation mechanism.

The framework adopted here does not impose any probabilistic structure on the evolution of $\{s_t\}$, but allows for such specifications as a special case. More importantly, it does not interpret regime dependence as parametric variation within a single mapping. Instead, the defining feature of the framework is the existence of multiple propagation laws, not the mechanism governing regime transitions.

\medskip

We formalize this framework as follows.

\begin{definition}\emph{(Regime-dependent system)}
A regime-dependent macroeconomic system is defined by a finite set of regimes $\mathcal{S}$ and a corresponding family of mappings $\{F_s\}_{s \in \mathcal{S}}$, where each $F_s : \mathbb{R}^n \to \mathbb{R}^n$ maps states into next-period states under regime $s$, and the evolution of the system is given by \eqref{eq:state_dynamics} for any realized sequence $\{s_t\} \subset \mathcal{S}$.
\end{definition}

Under this formulation, the evolution of the economy is determined by successive composition of regime-dependent operators. For a given initial condition $x_0$ and regime sequence $\{s_t\}$, the resulting trajectory is given by
\begin{equation}
\label{eq:DMR-FtF0}
x_{t+1} = F_{s_{t}} \circ \cdots \circ F_{s_0}(x_0).
\end{equation}
Each regime is thus associated with a distinct structural mapping, and macroeconomic dynamics arise from the ordered interaction of potentially heterogeneous laws of motion. Regimes are not interpreted as large realizations of shocks within a fixed system, nor as parameter perturbations within an invariant mapping. Instead, regime changes modify the mapping itself, altering the functional form that governs intertemporal propagation.

To analyze local propagation properties within each regime, we assume that each mapping $F_s$ is continuously differentiable. For each $s \in \mathcal{S}$, let $D F_s(x)$ denote the Jacobian matrix of $F_s$ at state $x$, capturing the local propagation structure associated with regime $s$.

While such local properties are defined at the level of individual operators, the overall dynamics of the system are governed by their ordered composition along the realized sequence $\{s_t\}$. As a result, properties of the global system cannot, in general, be inferred from the analysis of individual regimes in isolation.

This distinction is economically significant. Sequences of policy
regimes, institutional configurations, or financial environments may generate qualitatively different macroeconomic trajectories depending on their order, even when each regime, taken in isolation, exhibits stable
propagation properties. As illustrated in Section~3, compositions of individually stable regimes may fail to preserve stability when combined, reflecting the non-trivial interaction between distinct laws of motion.

Standard invariant-law macroeconomic models correspond to the special configuration in which
\begin{equation*}
F_s \equiv F \quad \text{for all } s \in \mathcal{S}.
\end{equation*}

In that case, economic evolution is governed by iteration of a single propagation operator. Structural invariance therefore appears as a degenerate configuration within the present architecture, in which operator heterogeneity is absent.

Once structural invariance is relaxed, global behavior is determined by ordered compositions of distinct operators and therefore by the realized
sequence of regimes. The core architectural distinction is thus between dynamics generated by iteration of a single mapping and dynamics generated by ordered composition across multiple, potentially non-commuting, structural laws of motion.

\section{Stability in Regime-Dependent Systems}
\label{sec:DMR-KillingExample}

The notion of stability in macroeconomic dynamics is traditionally
formulated with respect to a single invariant law of motion.
In that setting, stability is defined pointwise: given a mapping $x_{t+1} = F(x_t)$, a fixed point \(x^*\) satisfying
$$
F(x^*) = x^*
$$
is locally stable if trajectories starting sufficiently close remain close, and asymptotically stable if they, in addition, converge to \(x^*\).

In a regime-dependent system, however, the dynamics are governed by
a family of operators $\{F_s\}_{s \in \mathcal{S}}$, and trajectories are generated by ordered compositions given by \eqref{eq:DMR-FtF0}

As a result, stability can no longer be defined solely in terms of a single operator evaluated at a point. Instead, it is not a property of any operator evaluated at a point, but of the semigroup generated by $\{F_s\}$ along admissible sequences of regimes that the system may experience.

\medskip

To analyze stability properties, consider a common steady state $x^*$ such that $F_s(x^*) = x^*$ for all $s \in \mathcal{S}$. Let $U$ be a neighborhood of $x^*$ in which each mapping $F_s$ admits a first-order approximation. For each regime $s \in \mathcal{S}$, define the Jacobian matrix
\begin{equation*}
A_s = D F_s(x^*).
\end{equation*}
Then, for $x \in U$,
\begin{equation*}
F_s(x) = x^* + A_s (x - x^*) + o(\|x - x^*\|).
\end{equation*}
Unlike the invariant-law case, this linearization does not produce a single local operator, but a family $\{A_s\}_{s \in \mathcal{S}}$. As a result, even in an arbitrarily small neighborhood of $x^*$, the local dynamics are not governed by a fixed linear mapping. Instead, they are approximated by a sequence of regime-dependent linear operators:
\begin{equation*}
x_{t+1} - x^* \approx A_{s_t}(x_t - x^*), \quad x_t \in U.
\end{equation*}
Thus, local propagation is determined by ordered compositions of matrices drawn from $\{A_s\}$. Repeated composition along a regime sequence yields the local approximation
\begin{equation*}
x_{t+1} - x^* \approx A_{s_{t}} \cdots A_{s_0} (x_0 - x^*),
\end{equation*}
as long as the trajectory remains in $U$. This observation implies a fundamental departure from the standard analysis. Under an invariant law of motion, the Jacobian evaluated at $x^*$ fully characterizes local stability. Under structural regime change, the linearization yields not a single operator, but a family of operators whose interaction determines local behavior. Stability therefore ceases to be a pointwise property and becomes inherently sequential: it depends on the entire path of regimes, even in a local neighborhood.

\vspace{0.1in}

A natural object to characterize local stability of the linearized system is the joint spectral radius (JSR) of the set $\mathcal{A} = \{A_s : s \in \mathcal{S}\}$, defined as
$$
\rho(\mathcal{A}) := \limsup_{T \to \infty}
\sup_{(s_0,\dots,s_{T-1})}
\| A_{s_{T-1}} \cdots A_{s_0} \|^{1/T}.
$$
Importantly, the joint spectral radius is not a property of any single regime, but of the semigroup generated by the family $\{A_s\}$. It therefore captures
precisely the compositional nature of local dynamics under regime switching, rather than the stability properties of individual operators.

This object is standard in the analysis of switched and regime-dependent systems (see, e.g., \cite{Liberzon2003,Jungers2009}). If $\rho(\mathcal{A} < 1$, then the linearized system is uniformly exponentially stable under arbitrary switching, implying that
$$
\| x_t - x^* \| \leq M \alpha^t \| x_0 - x^* \|
$$
for some $M > 0,\ \alpha \in (0,1)$, and all admissible sequences $\{s_t\}$, provided the trajectory remains in $U$.

\medskip

It is crucial to emphasize that this characterization is inherently local. The matrices \(A_s\) arise from linearization at \(x^*\), and the joint spectral radius captures behavior only in a neighborhood of that point. It does not characterize the global dynamics of the nonlinear system, which depend on the full family of regime-specific operators \(\{F_s\}\) and their ordered compositions, beyond any local linear approximation.

\medskip

To illustrate these points, consider the following example, which isolates compositional effects in a simple macro-financial environment.

\newpage

\noindent \textbf{Example (Collateral values and borrowing capacity under\\ alternating financial regimes).}

\vspace{0.1in}

Let $x_t=(q_t,b_t)^\top$, where $q_t$ denotes the value of collateralizable assets and $b_t$ the effective borrowing capacity of the private sector. Consider an environment in which the direction of financial transmission depends on the prevailing regime.

In periods of normal financial conditions ($N$), abundant credit supports asset demand and sustains collateral values. In such environments, lending is primarily driven by balance-sheet capacity, funding conditions, and risk appetite, so that marginal changes in collateral values do not materially alter borrowing capacity. By contrast, under a constrained regime ($C$), borrowing capacity becomes more directly limited by the value of collateral, as lending conditions depend on the borrower’s balance-sheet position; when constraints bind, small changes in collateral values directly affect how much agents can borrow, while asset prices are largely driven by deleveraging and liquidity pressures rather than by marginal credit expansion. This asymmetry is standard in macro-financial analysis: in one regime, credit conditions support asset prices, while in the other, asset prices constrain credit.

\medskip

This difference can be summarized in terms of the relative strength of transmission channels. In normal conditions, borrowing capacity is relatively insensitive to marginal changes in collateral values, so that
$$
\frac{\partial b}{\partial q} \approx 0,
\qquad
\frac{\partial q}{\partial b} > 0,
$$
whereas under stressed conditions the pattern reverses: borrowing becomes tightly linked to collateral values, while the effect of credit on asset prices weakens,
$$
\frac{\partial b}{\partial q} > 0,
\qquad
\frac{\partial q}{\partial b} \approx 0.
$$
The preceding discussion motivates a regime-dependent specification in which the dominant transmission channel differs across regimes, capturing, in a stylized and regime-dependent form, balance-sheet effects and collateral constraints studied in the macro-financial literature (e.g., \cite{BernankeGertlerGilchrist1999, KiyotakiMoore1997}). Accordingly, the specification below emphasizes the dominant transmission channel in each regime.\footnote{The specification abstracts from secondary cross-effects within each regime in order to isolate the dominant direction of transmission. The qualitative result does not depend on this simplification.} Consider the following regime-specific laws of motion.

Under normal financial conditions ($N$),
\begin{equation}
\label{eq:collateral_N}
\begin{aligned}
q_{t+1} &= \alpha q_t
        + \mu\left(\frac{b_t}{1+b_t} - \frac{b^*}{1+b^*}\right)
        + \bar q, \\
b_{t+1} &= \beta b_t + \bar b,
\end{aligned}
\end{equation}
where $\alpha,\beta \in (0,1)$ capture intrinsic propagation, $\bar q>0$ and $\bar b>0$ are baseline components, and $\mu>0$ measures the effect of borrowing capacity on collateral values.

Under constrained regime ($C$),
\begin{equation}
\label{eq:collateral_C}
\begin{aligned}
q_{t+1} &= \alpha q_t + \bar q, \\
b_{t+1} &= \beta b_t
        + \nu\left(\frac{q_t}{1+q_t} - \frac{q^*}{1+q^*}\right)
        + \bar b,
\end{aligned}
\end{equation}
where $\nu>0$ measures the effect of collateral values on borrowing capacity.

The centered nonlinear terms ensure that both regimes share the same stationary reference point. A common fixed point $x^*=(q^*,b^*)$ is therefore defined by
$$
F_N(q^*,b^*)=(q^*,b^*), \qquad F_C(q^*,b^*)=(q^*,b^*).
$$

From \eqref{eq:collateral_N}--\eqref{eq:collateral_C}, it follows immediately that
$$
q^*=\frac{\bar q}{1-\alpha},
\qquad
b^*=\frac{\bar b}{1-\beta}.
$$

Thus, the model admits a common interior fixed point whenever $\alpha,\beta\in(0,1)$ and $\bar q,\bar b>0$.

\medskip

To analyze local stability, define the Jacobian matrices evaluated at the common fixed point:
\[
A_s = D F_s(x^*), \qquad s\in\{N,C\}.
\]

These are given by
$$
A_N=
\begin{pmatrix}
\alpha & \displaystyle \frac{\mu}{(1+b^*)^2} \\
0 & \beta
\end{pmatrix},
\qquad
A_C=
\begin{pmatrix}
\alpha & 0 \\
\displaystyle \frac{\nu}{(1+q^*)^2} & \beta
\end{pmatrix}.
$$

Each Jacobian is triangular. Hence their eigenvalues are simply $\alpha$ and $\beta$, so each regime is locally stable in isolation whenever $\alpha,\beta\in(0,1)$.

\medskip

For illustrative parameter values
$$
\alpha=0.8,
\qquad
\beta=0.8,
\qquad
\bar q=0.2,
\qquad
\bar b=0.2,
\qquad
\mu=1.6,
\qquad
\nu=1.6,
$$
the common fixed point is
$$
q^*=\frac{0.2}{1-0.8}=1,
\qquad
b^*=\frac{0.2}{1-0.8}=1.
$$

Substituting into the Jacobians yields
$$
A_N=
\begin{pmatrix}
0.8 & 0.4 \\
0 & 0.8
\end{pmatrix},
\qquad
A_C=
\begin{pmatrix}
0.8 & 0 \\
0.4 & 0.8
\end{pmatrix}.
$$

Their spectral radii satisfy
$$
\rho(A_N)=0.8<1,
\qquad
\rho(A_C)=0.8<1.
$$

Their ordered composition is
$$
A_C A_N=
\begin{pmatrix}
0.64 & 0.32 \\
0.32 & 0.80
\end{pmatrix}.
$$

The characteristic polynomial of $A_C A_N$ is
$$
\lambda^2 - 1.44\lambda + 0.4096 = 0,
$$
whose roots are
$$
\lambda_1 \approx 1.0498,
\qquad
\lambda_2 \approx 0.3902.
$$
Hence
$$
\rho(A_C A_N)\approx 1.0498 > 1.
$$

Therefore, although each regime is locally stable on its own, the two-period linearized dynamics generated by their ordered composition are locally unstable.

\medskip

The economic mechanism is transparent. Under regime $N$, borrowing capacity affects collateral values:
$$
b \longrightarrow q.
$$
Under regime $C$, collateral values affect borrowing capacity:
$$
q \longrightarrow b.
$$
The amplification mechanism arises only through the ordered composition
$$
b \;\xrightarrow{\,N\,}\; q \;\xrightarrow{\,C\,}\; b.
$$

Accordingly, the local dynamics under recurrent alternation are governed by
$$
x_{t+2}-x^* \approx (A_C A_N)(x_t-x^*),
$$
and therefore
$$
x_{t+2k}-x^* \approx (A_C A_N)^k (x_0-x^*).
$$
Since $\rho(A_C A_N)>1$, it follows that
$$
\|(A_C A_N)^k(x_0-x^*)\|\to\infty
\qquad \text{as } k\to\infty,
$$
for generic initial conditions sufficiently close to $x^*$.

\medskip

This example shows that regime-wise stability does not imply stability under regime switching. The instability is not a property of any regime considered in isolation, but arises from the ordered composition of structurally distinct propagation mechanisms. In this sense, local behavior is inherently compositional: it is determined not by the properties of individual operators, but by their interaction along admissible sequences of regimes.

This has a direct economic implication. An analytical framework based on a single invariant law of motion would conclude that the system is locally stable, since each regime admits a stable fixed point. This conclusion is misleading: instability arises precisely from the interaction between regimes, through feedback mechanisms that are not captured by regime-by-regime analysis. As a result, policy or risk assessments based on invariant-law reasoning may systematically underestimate the possibility of endogenous amplification and instability, even in environments where all observable regimes appear individually stable.

More generally, this type of mechanism arises whenever different regimes activate distinct directions of transmission between state variables, so that the associated propagation operators do not commute. In such environments, stability is not a regime-wise property, but a property of the regime-dependent propagation architecture. This type of configuration is typical in environments with balance-sheet effects that reverse direction across financial conditions

In the context of this example, the misdiagnosis has a concrete economic content. Since both borrowing capacity $b_t$ and collateral values $q_t$ are stable within each regime, a regime-by-regime assessment would classify leverage dynamics as locally stable. This overlooks the feedback loop generated under regime alternation. The interaction $b \rightarrow q \rightarrow b$ induces an amplification mechanism that is absent within any single regime, implying that leverage dynamics that appear stable in isolation may, in fact, become locally destabilizing once regime transitions are taken into account.

\section{Topological and Structural Separation from\\ Invariant-Law Systems}
\label{sec:structural_separation}

Section~3 showed that the behavior of a regime-dependent system under the composition of regimes may differ fundamentally from its behavior when each regime is analyzed in isolation. In particular, the example established that properties such as stability, which hold at the level of individual regimes, need not be preserved under their composition.

This feature is not a technical complication but a structural one. It suggests that regime-dependent systems may not, in general, be interpreted as perturbations or extensions of invariant-law systems. The relevant question is therefore not one of approximation, but of equivalence: whether both frameworks represent the same class of dynamical objects under different coordinates.

We now show that they are not. The separation is architectural and operates along two complementary dimensions. First, invariant-law systems and regime-dependent systems are not topologically equivalent: no system governed by multiple regimes can be transformed into a single-law system through a homeomorphism. Second, regime-dependent dynamics are dynamically irreducible: the regime index cannot be eliminated from the description of the dynamics without collapsing heterogeneous propagation mechanisms into a single operator

To formalize the first claim, we adopt topological conjugacy as the notion of equivalence and recall the standard definition of a discrete-time dynamical system.

\begin{definition}\em[Discrete-time dynamical system]
A discrete-time dynamical system is a pair $(X,T)$ where $X$ is a topological space and $T:X\to X$ is a continuous mapping. Given $x_0\in X$, the induced trajectory is
$$
x_t=T^t(x_0), \qquad t\ge 0.
$$
\end{definition}

Two systems $(X,T)$ and $(Y,S)$ are \emph{topologically conjugate} if there exists a homeomorphism $h:X\to Y$ such that
\begin{equation}
h\circ T = S\circ h.
\label{eq:conjugacy}
\end{equation}
This condition means that the two systems generate identical dynamics up to a homeomorphic transformation of the state space. Such a transformation can be interpreted as a change of coordinates: evolving the state under $T$ and then applying $h$ yields the same result as first transforming the state via $h$ and then evolving it under $S$.

We now formalize the two architectural classes introduced in Section~2.

\begin{definition}\em[Invariant-law class $\mathcal C$]
A system belongs to the class $\mathcal C$ if it is a discrete-time dynamical system of the form $(\mathbb{R}^n,F)$, where $\mathbb{R}^n$ is endowed with the Euclidean topology and $F:\mathbb{R}^n\to\mathbb{R}^n$ is a continuous mapping. The dynamics are given by
\begin{equation}
x_{t+1}=F(x_t).
\label{eq:single_law}
\end{equation}
\end{definition}

\begin{definition}\em[Regime-dependent class $\mathcal R$]
A system belongs to the class $\mathcal R$ if it is governed by a finite family of continuous operators $\{F_s\}_{s\in\mathcal S}$ with $|\mathcal S|\ge 2$, and its dynamics are given by
$$
x_{t+1}=F_{s_t}(x_t),
$$
for an exogenous sequence $\{s_t\}$.
\end{definition}

In invariant-law systems, trajectories are generated by repeated iteration of a single operator. In contrast, regime-dependent systems evolve through ordered compositions of distinct operators associated with the realized sequence of regimes.

\medskip

To establish the structural separation between invariant-law systems and regime-dependent systems, we introduce the following auxiliary topological space:
$$
\mathcal X = \bigsqcup_{s \in \mathcal S} \left(\mathbb{R}^n \times \{s\}\right),
$$
endowed with the product topology (Euclidean on $\mathbb{R}^n$ and discrete on $\mathcal S$).

\begin{theorem}\em[Topological non-equivalence]
\label{thm:non_equivalence}
If $|\mathcal S|\ge 2$, no system in $\mathcal R$ can be topologically conjugate to a system in $\mathcal C$.
\end{theorem}

\begin{proof}
Suppose, by contradiction, that there exists a system in $\mathcal C$ that is topologically conjugate to a regime-dependent architecture in $\mathcal R$. Then there exists a homeomorphism between $\mathcal X$ and $\mathbb{R}^n$.

However, $\mathbb{R}^n$ is connected, while the associated stratified space $\mathcal X$ is the disjoint union of the sets $\mathbb{R}^n\times\{s\}$, each connected, and therefore has $|\mathcal S|\ge 2$ connected components. Since the number of connected components is a topological invariant, $\mathcal X$ cannot be homeomorphic to $\mathbb{R}^n$. This contradicts the existence of a homeomorphism, and therefore rules out topological conjugacy.
\end{proof}

The result follows from a topological obstruction: invariant-law systems are defined on a connected state space, whereas $\mathcal X$ has $|\mathcal S| \ge 2$ connected components. Since the number of connected components is a topological invariant, $\mathcal X$ cannot be homeomorphic to $\mathbb{R}^n$, which rules out topological conjugacy.

Accordingly, regime-dependent dynamics cannot be represented as the iteration of a single invariant law of motion. Rather, they reflect the coexistence of multiple, structurally distinct transmission mechanisms whose interaction is inherently compositional and, in general, non-commutative.

A distinct but related question is whether the dependence on the regime can be eliminated dynamically. That is, even if the two classes are not topologically conjugate, one may ask whether regime-dependent dynamics can be represented by a single law of motion defined on a transformed state variable that does not depend on the regime.

Formally, the question is whether there exists a regime-blind reduction of the form
\begin{equation}
z_{t+1}=G(z_t),
\end{equation}
where $z=\phi(x)$ is independent of the regime index, such that the reduced dynamics replicate the regime-dependent evolution for all admissible realizations of the regime sequence.

To focus on non-degenerate representations that preserve the structure of the state, we restrict attention to injective mappings $\phi$. This ensures that the reduced state preserves the informational content of the original state, so that the reduction can be interpreted as a faithful representation rather than as a lossy aggregation. Under this requirement, the following result characterizes when a regime-blind reduction is possible.

\begin{theorem}\em[Dynamic irreducibility]
\label{the:DMR-irreducibilidad}
Suppose there exist maps $\phi : \mathbb{R}^n \to \mathbb{R}^n$ and $G : \mathbb{R}^n \to \mathbb{R}^n$ such that
\begin{equation}
\phi(F_s(x)) = G(\phi(x)) \quad \text{for all } x \in \mathbb{R}^n \text{ and all } s \in \mathcal S.
\end{equation}
Assume moreover that $\phi$ is injective. Then
$$
F_s = F_{s'} \quad \text{for all } s, s' \in \mathcal S.
$$
\end{theorem}

\begin{proof}
Fix arbitrary $s', s'' \in \mathcal S$. By assumption,
$$
\phi(F_{s'}(x)) = G(\phi(x)) = \phi(F_{s''}(x)) \quad \text{for all } x.
$$
Since $\phi$ is injective, it follows that $F_{s'}(x) = F_{s''}(x)$ for all $x$, hence $F_{s'} = F_{s''}$.
\end{proof}

Taken together, Theorem~\ref{thm:non_equivalence} and Theorem~\ref{the:DMR-irreducibilidad} establish an architectural separation between invariant-law and regime-dependent systems. The two classes are neither topologically equivalent nor reducible to one another under regime-blind representations that preserve the informational content of the state. In particular, regime-dependent dynamics cannot be represented by a single propagation law.

\vspace{0.1in}

The distinction is therefore not a matter of parametrization or approximation, but of dynamical architecture.

\vspace{0.1in}

These results already imply that representations of the form
$$
x_{t+1}=F(x_t)
$$
are fundamentally limited in environments where propagation mechanisms differ across regimes. Approaches that model the economy through a single propagation law—possibly augmented with shocks or time variation—rely on representations of the form $x_{t+1}=F(x_t)$. The analysis above shows that this representation is not innocuous when propagation is structurally heterogeneous. In particular, the limitation is not merely a matter of representation: it concerns the ability to represent the set of admissible transitions generated by regime-dependent dynamics, which arise from ordered compositions of propagation operators and, in general, are non-commutative.

\vspace{0.1in}

The next section sharpens this limitation by establishing a structural impossibility result: whenever regime operators fail to commute, no invariant-law representation can reproduce the set of regime-admissible trajectories generated by the system.

\section{Impossibility of Invariant-Law Representation}
\label{sec:DMR-ImposibilityTheorem}

Section \ref{sec:structural_separation} established that multiple-regime systems are structurally distinct from invariant-law systems at the level of state-space architecture and regime-blind reduction. In particular, no regime-blind transformation can eliminate the dependence on the active propagation operator without collapsing structural heterogeneity.

We now show that this structural separation extends further. Even abstracting from the regime sequence and focusing on the family of admissible state transitions induced on $\mathbb{R}^n$, the evolution generated by multiple propagation operators cannot, in general, be represented as iteration of a single map. More precisely, when regime-dependent operators fail to commute, the semigroup of admissible ordered compositions cannot be generated by iteration of any $F:\mathbb{R}^n\to\mathbb{R}^n$.

As illustrated in the example of Section \ref{sec:DMR-KillingExample}, non-commutativity arises naturally from regime-dependent transmission mechanisms rather than from parametric variation within a single law of motion. This observation has a direct algebraic implication: the set of admissible transitions is generated by ordered compositions of heterogeneous operators, forming a semigroup that need not be cyclic.

\vspace{0.2in}

\noindent\textbf{Semigroup generated by regime operators.}

Let
$$
\mathcal{G}
=
\langle F_s : s\in\mathcal{S}\rangle
$$
denote the semigroup generated by finite compositions of the family
$\{F_s\}_{s\in\mathcal{S}}$.
That is,
$$
\mathcal{G}
=
\left\{
F_{s_k}\circ \cdots \circ F_{s_1}
:\;
k\ge 1,\; s_i\in\mathcal{S}
\right\}.
$$

Elements of $\mathcal{G}$ correspond to all admissible ordered compositions
of regime-dependent propagation operators. In particular, for any regime
sequence $\{s_t\}_{t\ge 0}$, the state at time $t$ satisfies
$$
x_t
=
F_{s_{t-1}}\circ \cdots \circ F_{s_0}(x_0),
$$
so the dynamic evolution of DMR is generated by ordered compositions of elements of $\mathcal{G}$, indexed by regime sequences.

Thus, reproducing the full set of regime-admissible trajectories requires reproducing all admissible compositions of regime operators. In particular, reproducing the full set of regime-admissible trajectories requires that the semigroup of iterates generated by $F$ coincides with $\mathcal{G}$.

\vspace{0.2in}

\noindent\textbf{Invariant-law semigroup.}

Given a single map $F:\mathbb{R}^n\to\mathbb{R}^n$, the semigroup generated by iteration is
$$
\langle F\rangle
=
\{F^k : k\ge 1\}.
$$
This is a cyclic semigroup: every element is a power of a single generator.

\vspace{0.1in}

The central structural question is therefore:

\begin{quote}
Can the semigroup $\mathcal{G}$ generated by multiple regime operators
coincide with a cyclic semigroup $\langle F\rangle$ generated by
a single invariant operator?
\end{quote}

If this were possible, then the entire dynamic architecture of DMR
could be represented by iteration of one map. In what follows, we show that this cannot occur whenever the regime operators fail to commute.

\vspace{0.1in}

\begin{lemma}\em[Cyclic semigroups are commutative]
Let $F:\mathbb{R}^n\to\mathbb{R}^n$.
Then for all integers $m,n\ge 1$,
$$
F^m\circ F^n
=
F^{m+n}
=
F^n\circ F^m.
$$
Hence every semigroup of the form $\langle F\rangle$ is commutative.
\end{lemma}

\begin{proof}
By associativity of composition and the definition of iteration,
$$
F^m\circ F^n
=
F^{m+n}
=
F^n\circ F^m.
$$
\end{proof}

\vspace{0.1in}

Thus, if $\mathcal{G}=\langle F\rangle$ for some map $F$,
then $\mathcal{G}$ must be commutative. This imposes a strong algebraic
restriction: any dynamic system generated by iteration of a single map necessarily admits commutative composition of operators.

\begin{theorem}\em[Impossibility of Invariant-Law Representation]
\label{thm:impossibility-strong}
Suppose there exist distinct regimes $s'\neq s''\in\mathcal{S}$
such that the operators $F_{s'}$ and $F_{s''}$ do not commute. Then there exists no map
$F:\mathbb{R}^n\to\mathbb{R}^n$
such that
$$
\mathcal{G} = \langle F\rangle.
$$
In particular, there exists no invariant-law representation $x_{t+1}=F(x_t)$ whose iterates reproduce all regime-admissible trajectories generated by $\{F_s\}_{s\in\mathcal{S}}$.
\end{theorem}

\begin{proof}
Since there exist $s', s''\in\mathcal{S}$ such that $F_{s'}$ and $F_{s''}$ do not commute, there exists $x\in\mathbb{R}^n$
such that
\begin{equation}
\label{eq:DMR-Fsx}
F_{s'}\circ F_{s''}(x)\neq F_{s''}\circ F_{s'}(x).
\end{equation}
Now, assume, for contradiction, that there exists a map $F:\mathbb{R}^n\to\mathbb{R}^n$
such that
$$
\mathcal{G}=\langle F\rangle.
$$

Since $\mathcal{G}$ is generated by the family $\{F_s\}_{s\in\mathcal{S}}$,
we have $F_s\in \mathcal{G}$ for every $s\in\mathcal{S}$. If
$\mathcal{G}=\langle F\rangle$, it follows that $F_s\in \langle F\rangle$ for every $s\in\mathcal{S}$. Hence there exists an integer
$m_s\in\mathbb{N}$ such that
$$
F_s = F^{m_s}.
$$

It follows that for $s', s''\in\mathcal{S}$,
$$
F_{s'}\circ F_{s''} = F^{m_{s'}}\circ F^{m_{s''}} = F^{m_{s'}+m_{s''}} = F^{m_{s''}}\circ F^{m_{s'}} = F_{s''}\circ F_{s'},
$$
so for all $x\in\mathbb{R}^n$ one has
$$
F_{s'}\circ F_{s''}(x) = F_{s''}\circ F_{s'}(x).
$$
This contradicts \eqref{eq:DMR-Fsx}, which establishes non-commutativity at $x$. Therefore, no such map $F$ can exist.
\end{proof}

Theorem~\ref{thm:impossibility-strong} shows that the structural distinction between invariant-law and regime-dependent systems extends to the level of dynamic representation. Invariant-law systems generate evolution through iteration of a single operator, inducing a cyclic and necessarily commutative semigroup of transitions. By contrast, regime-dependent systems generate evolution through ordered compositions of structurally distinct operators. Whenever these operators fail to commute, the order of composition becomes economically and dynamically relevant, and the resulting semigroup cannot be generated to iteration of a single map.

This limitation is structural. It does not arise from functional form assumptions or from the choice of state variables, but from the algebraic organization of admissible transitions. In particular, the dependence of outcomes on the sequence of regimes—captured by non-commutativity—cannot be encoded within any invariant propagation law.

As a consequence, any representation of the form
$$
x_{t+1} = F(x_t)
$$
fails to capture a fundamental dimension of the dynamics whenever regime-dependent propagation mechanisms are present. In such environments, invariant-law formulations do not merely approximate the underlying system but misrepresent its dynamic structure by suppressing the role of ordered composition.

\section{Conclusion}
\label{sec:DMR-Conclusion}

This paper introduced Dynamic Macroeconomics with Multiple Regimes (DMR) as a structural reformulation of dynamic macroeconomic analysis. The central departure is architectural. Instead of representing the economy as evolving under a single continuous law of motion defined on a connected state space, DMR treats economic evolution as governed by multiple structurally distinct propagation operators acting on $\mathbb{R}^n$, whose ordered composition is indexed by regime sequences.

This implies that the use of invariant-law representations entails a structural restriction: it excludes a class of admissible economic dynamics that arise under regime-dependent propagation. As a result, analyses conducted within an invariant-law framework necessarily fail to capture qualitative features of economic evolution that are intrinsic to environments with structural change.

Three structural results were established. First, invariant-law systems and regime-dependent systems belong to distinct topological conjugacy classes, reflecting a fundamental difference in state-space architecture. Second, the regime coordinate is dynamically irreducible: it cannot be eliminated through continuous reduction without collapsing heterogeneous propagation mechanisms. Third, whenever regime operators fail to commute, the semigroup generated by their ordered compositions cannot be cyclic; as a consequence, no invariant-law representation of the form $x_{t+1}=F(x_t)$ can reproduce the full set of regime-admissible trajectories through the iteration of a single invariant law of motion.

Taken together, these results establish a sharp structural separation between invariant-law and regime-dependent systems. Regime dependence is not a refinement of a unified law of motion, but a fundamentally different dynamic architecture. Once structural multiplicity is admitted, the economy can no longer be represented as the iteration of a single mapping. Instead, its evolution is governed by ordered compositions of heterogeneous operators, indexed by regime sequences.

In particular, the impossibility result identifies a fundamental restriction in standard macroeconomic reasoning. Conventional frameworks rely on the assumption that economic dynamics can be represented by a single invariant law of motion. The analysis above shows that this restriction is not without loss of generality. Whenever propagation operators fail to commute, there exists no mapping $F:\mathbb{R}^n\to\mathbb{R}^n$ capable of reproducing the set of regime-admissible trajectories. This limitation is structural: it does not arise from approximation or lack of flexibility within a given model class, but from the impossibility of reducing ordered compositions of heterogeneous operators to the iteration of a single map. This limitation arises because such compositions are, in general, non-commutative.

This shift has direct implications for stability analysis. In invariant-law systems, stability is defined with respect to a single propagation mechanism. Under regime dependence, stability is a property of families of operators and their admissible compositions. Dynamics that are stable within each regime in isolation may generate instability when composed across regimes. Stability must therefore be understood as a property of the interaction structure among propagation mechanisms.

The implications extend to policy analysis. In a single-law environment, stabilization is defined relative to the control of a unique propagation mechanism. Under structural multiplicity, policy must be evaluated with respect to admissible regime sequences and their induced compositions. There is no unified mapping relative to which stabilization can be defined. Instruments that stabilize the economy within a given regime do not necessarily prevent instability arising from the interaction of distinct propagation mechanisms. Policy effectiveness is therefore inherently regime-contingent.

More broadly, the distinction developed in this paper concerns the structure of the law of motion itself. In environments subject to structural regime
change, the relevant dynamic object is operator composition rather than iteration. This shift defines a broader class of admissible macroeconomic systems, but more importantly, it establishes that invariant-law models constitute a structurally restricted subset within that class.

Dynamic Macroeconomics with Multiple Regimes provides a framework in which structural change is not treated as a perturbation of a fixed propagation mechanism, but as an intrinsic component of macroeconomic dynamics. In this sense, the contribution of the paper is not to extend existing models, but to identify a structural limitation of invariant-law representations and to provide the foundations of an alternative dynamic architecture.

\section*{Research Agenda}
\addcontentsline{toc}{section}{Research Agenda}

The framework developed in this paper provides the structural
infrastructure for a broader research program on regime-dependent
macroeconomic dynamics. Several directions follow naturally from
the architecture established here, each corresponding to a class
of macroeconomic phenomena for which invariant-law representations
are structurally insufficient.

A first direction concerns macro-financial stability and the
endogenous emergence of crises. When financial leverage governs
the activation of distinct propagation operators, the boundary
separating stable trajectories from those triggering a regime
transition becomes a dynamic object generated endogenously by
the system, rather than an exogenously imposed constraint. The
analysis of this fragility frontier provides a structural foundation
for macro-prudential policy.

A second direction concerns financial collapse and recovery.
Once the economy enters a regime governed by a distinct
propagation operator, recovery cannot be characterized as
asymptotic behavior under the prevailing dynamics, but must
be formulated as a controlled reachability problem in which
policy alters the effective law of motion. This perspective
yields a structural characterization of crisis management,
including threshold conditions for effective intervention
and trade-offs between the cost and the speed of recovery.

A third direction concerns optimal policy under endogenous
structural transition risk. When fragility raises the probability
of transition toward a low-efficiency regime and continuation
values differ across regimes, the planner's optimality condition
acquires a structural component absent from invariant-law
benchmarks. This mechanism provides a structural reinterpretation
of preventive policy, persistent low growth, and poverty traps
as consequences of regime-dependent propagation rather than
of local non-convexities within a single accumulation operator.

A fourth direction concerns sovereign risk and debt dynamics
under regime-dependent propagation, where the transmission
between fiscal variables, financing conditions, and macroeconomic
outcomes may differ across structural environments. A fifth
direction concerns monetary policy transmission under structural
regime change, where the mapping from policy instruments to
macroeconomic outcomes is itself regime-specific rather than
parametrically variable within a unified mechanism.

These directions share a common structural feature: in each
case, the relevant economic phenomena cannot be represented
as the iteration of a single invariant law of motion, but
arise from the ordered composition of heterogeneous propagation
operators. The framework developed here provides the formal
infrastructure within which such phenomena admit a unified
analytical treatment. Several of these directions are developed
in companion work.

\bibliography{RDM_references}

\end{document}